\documentclass[12pt,envcountsame,envcountsect]{llncs}
\usepackage{ifpdf}
\usepackage{a4wide,amssymb,amsmath,enumerate,hhline}
\title{An optimal construction of Hanf sentences}
\author{Benedikt Bollig\inst{1} \and Dietrich Kuske\inst{2}}

\institute{Laboratoire Sp\'ecification et V\'erification, \'Ecole
  Normale Sup\'erieure de Cachan \& Centre National de la Recherche
  Scientifique \and Institut f\"ur Theoretische Informatik, TU Ilmenau}

\def\bN{\mathbb N}
\def\cA{\mathcal A}
\def\cB{\mathcal B}
\def\dist{\mathrm{dist}}
\def\sphere{\mathrm{sph}}
\begin{document}
\maketitle

\begin{abstract}
  We give a new construction of formulas in Hanf normal form that are
  equivalent to first-order formulas over structures of bounded
  degree. This is the first algorithm whose running time is shown to
  be elementary. The triply exponential upper bound is complemented by
  a matching lower bound.
\end{abstract}

\section{Introduction}

Various syntactical normal forms for semantical properties of
structures are known. For example, every first-order definable
property that is preserved under extensions of structures is definable
by an existential first-order sentence (\L{}o\'s-Tarski
\cite{Tar54,Los55}). Gaifman's normal form is another example that
formalizes the observation that first-order logic can only express
local properties \cite{Gai82}. A third example in this line is Hanf's
theorem, giving another formalization of locality of first-order logic
(at least for structures of bounded degree) \cite{Han65,EbbF95}.

Gaifman's and Hanf's theorems have found applications in finite model
theory and in particular in parametrized complexity. Namely, they lead
to efficient parametrized algorithms deciding whether a formula holds
in a (finite) structure
\cite{See96,Lib00,FriG01,FriG04,KusL05a,DawGKS06,Lin08,KusL08,KusL11}
and even for more general algorithms that list all the satisfying
assignments \cite{DurG07,KazS11}. Hanf's theorem was also used in the
transformation of logical formulas into different automata models
\cite{Tho97b,GiaRST96,BolL04,BolK08}.

In \cite{DawGKS07}, it was shown that passing from arbitrary formulas
to those in \L{}o\'s-Tarski or Gaifman normal form leads to a
non-elementary blowup. The same paper also proves that for structures
of bounded degree, the blowup for Gaifman's normal form is between 2-
and 4-fold exponential, and that for \L{}o\'s-Tarski normal forms (for
a restricted class of structures) is between 2- and 5-fold
exponential.

This paper shows that Hanf's normal form can be computed in three-fold
exponential time and that this is optimal since there is a necessary
blowup of three exponentials when passing from general first-order
formulas to their Hanf normal form. We remark (as already observed by
Seese~\cite{See96}) that the first construction of Hanf normal
forms~\cite{EbbF95} is not effective since satisfiability of
first-order formulas in graphs of bounded degree is
undecidable, also when we restrict to finite structures~\cite{Wil94}. Only Seese \cite{See96} gave a small
additional argument showing that Hanf normal forms can indeed be
computed. But his algorithm is not primitive recursive. This was
improved later to a primitive-recursive algorithm by Durand and
Grandjean \cite{DurG07} and (independently) by
Lindell~\cite{Lin08}. Their papers do not give an upper bound for the
construction of Hanf normal forms, but on the face of it, the
algorithm seems not to be elementary.\footnote{In the meantime, A.\
  Durand has informed us of ongoing work aiming at an elementary upper
  bound for their algorithm.} Their algorithm is a
quantifier-alternation procedure that only works if the signature
consists of finitely many injective functions (following Seese, one
can bi-interprete every structure of bounded degree in such a
structure, so this is no real restriction of the
algorithm). Differently, our algorithm follows the original proof of
Hanf's theorem very closely by examining spheres of bounded diameter,
but avoiding the detour via Ehrenfeucht-Fra\"\i ss\'e-games.

\paragraph{Acknowledgement} We would like to thank Luc Segoufin and
Arnaud Durand for comments on an earlier version and for hints to the
literature that improved this paper.

\section{Definitions and background}

Throughout this paper, let $L$ be a finite relational signature and
let $L_m$ denote the extension of $L$ by the constants $c_1,
c_2,\dots, c_m$. Let $\cA$ be an $L_m$-structure. We write $a\in\cA$
when we mean that $a$ is an element of the universe of
$\cA$. Furthermore, $\bar a$ denotes a tuple $(a_1,\dots,a_n)$ of
length $n$ of elements of some structure $\cA$ and $\bar x$ is the
list of variables $(x_1,\dots,x_n)$. In both cases, $n$ will be
determined by the context. Finally, we define a distance (from
$\bN\cup\{\infty\}$) on the universe of $\cA$ setting
$\dist^\cA(a,b)=0$ iff $a=b$ and $\dist^\cA(a,c)=d+1$ if there exists
$b\in\cA$ with $\dist^\cA(a,b)\le d$, there is some tuple in some of
the relations of $\cA$ that contains both, $b$ and $c$, and there is
no such $b\in\cA$ with $\dist^\cA(a,b)<d$, and $\dist^\cA(a,b)=\infty$
if $\dist^\cA(a,b)\neq d$ for all $d\in\bN$. Next, the \emph{degree}
of $a\in\cA$ is the number of elements $b\in\cA$ with
$\dist^\cA(a,b)=1$, the degree of $\cA$ is the supremum of the degrees
of $a\in\cA$.

Let $\cA$ be an $L$-structure, $\bar a=(a_1,\dots,a_n)\in\cA$,
and $d>0$. Then $B_d^\cA(\bar a)$ is the set of elements $b\in\cA$
with $\dist^{\cA}(a_i,b)<d$ for some $1\le i\le n$. The
\emph{$d$-sphere around $\bar a$} is the $L_n$-structure
\[
    S_d^\cA(\bar a)=(\cA\restriction B_d^\cA(\bar a),\bar a)\,.
\]
A \emph{$d$-sphere (with $n$ centers)} is an $L_n$-structure
$(\cA,\bar a)$ with $B_d^{\cA}(\bar a)=\cA$. The $L_n$-structure
$(\cA,\bar a)$ is a \emph{sphere} if there exists $d>0$ such that
$(\cA,\bar a)$ is a $d$-sphere; the least such $d$ is denoted
$d(\tau)$ and is the \emph{radius of $(\cA,\bar a)$}. The $d$-sphere
$\tau$ is \emph{realised by $\bar a$ in $\cA$} if
\[
  \tau\cong S_d^\cA(\bar a)\,.
\]

If two $L$-structures $\cA$ and $\cB$ satisfy exactly the same
first-order sentences, then we write $\cA\equiv\cB$. If they only
satisfy the same sentences of quantifier rank $\le r$, then
$\cA\equiv_r\cB$. Provided the degrees of $\cA$ and $\cB$ are finite,
both these concepts can be characterized using the number of
realisations of spheres. 

\begin{theorem}[Hanf \cite{Han65}]
  \label{T Hanf}
  For any $L$-structures $\cA$ and $\cB$, we have $\cA\equiv\cB$
  whenever any sphere in $\cA$ or $\cB$ is finite and any sphere is
  realised in $\cA$ and $\cB$ the same number of times or
  $\ge\aleph_0$ times.
\end{theorem}

This result was sharpened by Fagin, Stockmeyer \& Vardi (see also
Ebbinghaus \& Flum \cite{EbbF95}) to characterize the relation~$\equiv_r$:

\begin{theorem}[Fagin et al.~\cite{FagSV93}]
  \label{T Hanf 2}
  For all $r,f\in\mathbb N$ there exist $d,m\in\mathbb N$ (where $d$
  depends on $r$, only) such that for any $L$-structures $\cA$ and
  $\cB$ of degree $\le f$, we have $\cA\equiv_r\cB$ whenever any
  $d$-sphere with one center is realised in $\cA$ and $\cB$ the same
  number of times or $\ge m$ times.
\end{theorem}

\begin{proof}[of both theorems]
  The proof proceeds by showing that the respective counting property
  implies that duplicator has a winning strategy in the
  Ehrenfeucht-Fra\"\i ss\'e-game \cite{EbbF95,Lib04}. This then
  implies the respective equivalence of $\cA$ and $\cB$.\qed
\end{proof}

This theorem has (at least) three different applications: The first
application (and its original motivation in~\cite{FagSV93}) is a
technique to prove that certain properties $\mathfrak P$ are not
expressible in first-order logic: One provides two lists of structures
$\cA_r$ and $\cB_r$ where $\cA_r$ has the desired property and $\cB_r$
does not. Furthermore, for any $r$, $\cA_r$ and $\cB_r$ satisfy the
counting condition from Theorem~\ref{T Hanf 2} with $d$ and $m$
determined by $r$ and the degree $f$ of $\cA_r$ and $\cB_r$. This
implies $\cA_r\equiv_r\cB_r$ and therefore the property $\mathfrak P$
cannot be expressed by a first-order sentence of quantifier depth
$r$. Since this holds for all $r$, the property is not first-order
expressible. The simplest such property is connectivity of a graph
where $\cA_r$ can be chosen a circle of size $\max(m,2d)$ and $\cB_r$
a disjoint union of two copies of $\cA_r$ ($m$ and $d$ are the
constants from Theorem~\ref{T Hanf 2} for $f=2$).

The second application is an efficient evaluation of first-order
properties on finite structures of bounded degree
\cite{See96,FriG04,DurG07}: The idea is to count the number of
realisations of spheres up to the threshold $m$ and, depending on the
vector obtained that way, decide whether the formula holds or not (we
will come back to this aspect later in this section).

The third application is a normal form for first-order
sentences~\cite{EbbF95}. For a finite $d$-sphere~$\tau$ with $n$
centers, let $\sphere_\tau(\bar x)$ denote a formula such that
$(\cA,\bar a)\models\sphere_\tau$ iff $S_d^\cA(\bar a)\cong\tau$. A
\emph{Hanf sentence} asserts that there are at least $m$ realisations
of the finite sphere $\tau$ with one center. Formally, it has the form
\begin{equation}
   \exists x_1,x_2,\dots, x_{m}:
    \bigwedge_{1\le i<j\le m}x_i\neq x_j\land
    \forall x:\left(\left(\bigvee_{1\le i\le m} x=x_i\right)
                            \to \sphere_\tau(x)\right)
    \label{eq:shorthand}
\end{equation}
which we abbreviate as
\[
   \exists^{\ge m}x:\sphere_\tau(x)\,.
\]
A sentence is in
\emph{Hanf normal form} if it is a Boolean combination of Hanf
sentences.

Let $\varphi$ and $\psi$ be two formulas with free variables in
$x_1,\dots,x_n$. To simplify notation, we will say that $\varphi$ and
$\psi$ are \emph{$f$-equivalent} if, for all structures $\cA$ of
degree~$\le f$, we have
\[
 \cA\models\forall x_1\forall x_2\dots\forall x_n:
   (\varphi\leftrightarrow\psi)\,.
\]

\begin{corollary}[Ebbinghaus \& Flum \cite{EbbF95}]
  \label{C Hanf}
  For every sentence $\varphi$ and all $f\in\mathbb N$, there exists
  an $f$-equivalent sentence $\psi$ in Hanf normal form.
\end{corollary}

\begin{proof}
  Let $r$ be the quantifier rank of $\varphi$ and let $d$ and $m$
  denote the numbers from Theorem~\ref{T Hanf 2}. Then there are only
  finitely many $d$-spheres of degree $\le f$ with one center; let
  $(\tau_1,\dots,\tau_n)$ be the list of these spheres. Now we
  associate with every structure $\cA$ of degree~$\le f$ a tuple
  $t^\cA\in\{0,1,\dots,m\}^n$ as follows: For $1\le i\le n$, let
  $t^\cA_i$ denote the minimum of $m$ and the number of $a\in\cA$ with
  $S^{\cA}_d(a)\cong\tau_i$. Note that there are only finitely many
  tuples $t^\cA$. Now $\psi$ is a disjunction. It has one disjunct for
  every $t\in\{0,1,\dots,m\}^n$ for which there exists a structure
  $\cA$ of degree~$\le f$ with $\cA\models\varphi$ and $t=t^\cA$. This
  disjunct is the conjunction of the following formulas for $1\le i\le
  n$:
  \[
  \begin{cases}
    \exists^{=t_i}x:\sphere_\tau(x) & \text{ if }t_i< m\\
    \exists^{\ge m}x:\sphere_\tau(x) & \text{ if }t_i= m\,.
  \end{cases}
  \]
  \qed
\end{proof}

Note that $\varphi$ is satisfiable if and only if the disjunction
$\psi$ is not empty.\label{Page-non-effective} Hence an effective
construction of $\psi$ would allow us to decide satisfiability of
first-order formulas in structures of degree $\le f$ which is not
possible \cite{Wil94}.

We now turn to finite structures. Clearly, the disjunction $\psi$ as
in the above corollary is also equivalent to $\varphi$ for all finite
structures of degree $\le f$. But in this context, we can also define
another disjunction $\psi_{\mathrm{fin}}$ by taking only those
$t\in\{0,1,\dots,m\}^n$ for which there exists a \emph{finite}
structure $\cA$ of degree $\le f$ with $\cA\models\varphi$ and
$t=t^\cA$ (cf., e.g., \cite[page 101]{Lib04}). As above, an effective
construction of $\psi_{\mathrm{fin}}$ would allow us to decide
satisfiability of first-order formulas in \emph{finite} structures of
degree $\le f$ which, again, is not possible \cite{Wil94}.

Despite the fact that the proof of Cor.~\ref{C Hanf} is not
constructive, Seese showed that some sentence $\psi$ as required in
Cor.~\ref{C Hanf} can be computed.

\begin{theorem}[Seese~\protect{\cite[page 523]{See96}}]
  \label{T Seese}
  From a sentence $\varphi$ and $f\in\mathbb N$, one can compute an
  $f$-equivalent sentence in Hanf normal form.
\end{theorem}

\begin{proof}
  Let $\beta$ express that a structure has degree $\le f$. Then search
  for a tautology of the form $\beta\to(\varphi\leftrightarrow\psi)$
  where $\psi$ is a sentence in Hanf normal form. Since the set of
  tautologies is recursively enumerable, we can do this search
  effectively. And since we know from Theorem~\ref{T Hanf 2} that an
  $f$-equivalent sentence in Hanf normal form exists, this search will
  eventually terminate successfully.\qed
\end{proof}

Note that Seese's procedure to compute $\psi$ is not primitive
recursive. A primitive recursive construction of a Hanf normal form
was described by Durand and Grandjean \cite{DurG07} and independently
by Lindell~\cite{Lin08}. They present a quantifier elimination scheme
and do not rest their reasoning on Ehrenfeucht-Fra\"iss\'e-games. But
so far, no elementary upper bound for the running time of their
algorithm is known. The main result of this paper is an elementary
procedure for the computation of a Hanf normal form. This is achived
by a new (direct) proof of Corollary~\ref{C Hanf} that does not use
games.

The effective constructions of Hanf normal forms led
Seese~\cite{See96}, Durand and Grandjean~\cite{DurG07} and
Lindell~\cite{Lin08} to efficient algorithms for the evaluation of
first-order queries on structures of bounded degree. Seese showed that
sentences in Hanf normal form can be evaluated in time linear in the
structure and the Hanf normal form. Consequently, the set of pairs
$(\mathcal A,\varphi)$ with $\mathcal A$ a structure of degree $\le f$
and $\varphi$ a sentence with $\mathcal A\models\varphi$ can be
decided in time
\begin{equation}
  g_1(|\varphi|,f)+g_2(|\varphi|,f)\cdot|\mathcal A|\label{eq:Seese}
\end{equation}
where $g_1(|\varphi|,f)$ is the time needed to compute the Hanf normal
form and $g_2(|\varphi|,f)$ is its size\footnote{It should be noted
  that Frick and Grohe proved this problem to be solvable with $g_1$
  the identity and $g_2$ triply exponential in $|\varphi|$ and $f$
  \cite{FriG04}.}  (it can be shown that the function $g_2$ is
elementary since the radiuses appearing in the Hanf normal form can be
bound). Since Seese's construction is not primitive recursive, the
function $g_1$ is not primitive recursive. The constructions by Durand
and Grandjean and by Lindell show that $g_1$ can be replaced by a
primitive recursive function~$g_1'$. Since they get another Hanf
normal form, also the function $g_2$ changes to $g_2'$, say (but as
for Seese's Hanf normal form, also this function is elementary).

In addition, Durand and Grandjean and Lindell show that the set of
tuples $\bar a$ from $\cA$ with $\cA\models\varphi(\bar a)$ can be
computed in time
\begin{equation}
  g_1'(|\varphi|,f)+g_2'(|\varphi|,f)\cdot(|\cA|+|\{\bar a\mid
\cA\models\varphi(\bar a)\}|)\label{eq:Durand}
\end{equation}
where $\varphi$ is a first-order formula and $f$ is the degree of the
structure $\cA$. Since $g_2'$ is elementary, this result is of
particular importance if this set has to be computed for many
structures~$\mathcal A$ and a fixed formula~$\varphi$.  This was
recently improved by Kazana and Segoufin who compute this set in time
\[
   2^{2^{2^{O(|\varphi|)}}}\cdot
     (|\cA|+|\{\bar a\mid \cA\models\varphi(\bar a)\}|)\,.
\]
Here, the triply exponential factor originates from the work by Frick
and Grohe \cite{FriG04} and the summand $g_1'$ is avoided since they
do not precompute a Hanf normal form.  Our result in this paper will
show that the Hanf normal form can be computed in triply exponential
time. Consequently, the functions from (\ref{eq:Seese}) and from
(\ref{eq:Durand}) can be replaced by triply exponential functions. As
a result, the model checking algorithm by Seese and the enumeration
algorithm by Durand and Grandjean and by Lindell perform as well as
the algorithms by Frick and Grohe and by Kazana and Segoufin, resp.

\section{Construction of a Hanf normal form}

A \emph{Hanf formula with free variables from $x_1,\dots,x_n$} is a
formula of the form
\[
   \exists^{\ge m}y:\sphere_\tau(\bar x,y)
\]
where $\tau$ is a sphere with $n+1$ centers. A formula is in
\emph{Hanf normal form} if it is a Boolean combination of Hanf
formulas.

\begin{theorem}\label{T-construction}
  From a formula $\Phi$ with free variables among $\bar x$ and
  $f\ge 1$, one can construct an $f$-equivalent formula $\Psi$ in
  Hanf normal form. This construction can be carried out in time
  \[
    2^{f^{2^{O(|\Phi|)}}}\,.
  \]
\end{theorem}

The construction of $\Psi$ from $\Phi$ will be done by induction on
the construction of $\Phi$. The central part in this induction is
described by the following lemma (the proof of
Theorem~\ref{T-construction} can be found at the end of this section).

\begin{lemma}\label{L-construction}
  From a formula $\varphi$ in Hanf normal form with free variables
  among $\bar x,x_{n+1}$ and $f\ge 1$, one can construct a formula
  $\psi$ in Hanf normal form with free variables in $\bar x$ such that
  $\exists x_{n+1}:\varphi$ and $\psi$ are $f$-equivalent. This
  construction can be carried out in time $|\varphi|\cdot
  2^{n^{O(1)}\cdot f^{O(d)}}$ where $d$ is the maximal radius of a
  sphere appearing in~$\varphi$. Furthermore, the largest radius
  appearing in $\psi$ is~$3d$.
\end{lemma}

\begin{proof}
  Set $e=3d$. The formula $\psi$ will be a disjunction with
  one disjunct for every $e$-sphere $\tau'$ with $n+1$ centers. This
  disjunct will have the form
  \[
    \psi_{\tau'}=\varphi_{\tau'}\land\exists^{\ge1}x_{n+1}:\sphere_{\tau'}\,.
  \]
  We next describe how $\varphi_{\tau'}$ is obtained from
  $\varphi$. For this, let $\alpha=\exists^{\ge
    m}x_{n+2}:\sphere_\tau$ be some Hanf formula appearing in
  $\varphi$. This formula will be replaced by the Hanf formula
  $\alpha'$ that we construct next. In this construction, we
  distinguish two cases, namely whether the $d(\tau)$-sphere around
  $c_{n+1}c_{n+2}$ in $\tau$ is connected or not.
  \begin{enumerate}[(a)]
  \item $S^{\tau}_{d(\tau)}(c_{n+1}c_{n+2})$ is connected.

    Let $p$ denote the number of elements $c\in
    B^{\tau'}_{2d(\tau)}(c_{n+1})$ with
    \[
      S^{\tau'}_{d}(\bar c c_{n+1} c)\cong\tau
    \]
    and set
    \[
      \alpha'=
      \begin{cases}
        \top & \text{ if }p\ge m\\
        \bot & \text{ otherwise.}
      \end{cases}
    \]

  \item $S^{\tau}_{d(\tau)}(c_{n+1}c_{n+2})$ is not connected.

    Let
    \[
       \sigma=S_{d(\tau)}^\tau(\bar c c_{n+2})
    \]
    and write $p$ for the number of $c\in
    B^{\tau'}_{2d(\tau)}(c_{n+1})$ with
    \[
       S^{\tau'}_{d(\tau)}(\bar c c)\cong\sigma\,.
    \]
    In this case, set 
    \[
       \alpha'= \exists^{\ge m+p} x_{n+2}:\sphere_\sigma(\bar x,x_{n+2}) \,.
    \]
  \end{enumerate}

  This finishes the construction of $\varphi_{\tau'}$ and therefore of
  the disjunction $\psi$. Clearly, $\psi$ is in Hanf normal form.

  Now let $a_{n+1}\in\cA$ with $S^\cA_d(\bar a a_{n+1})\cong\tau'$. We
  will show
  \[
     (\cA,\bar a a_{n+1})\models\alpha\iff(\cA,\bar a)\models\alpha'
  \]
  again distinguishing the two cases above.
  \begin{enumerate}[(a)]
  \item First let $S^{\tau}_{d(\tau)}(c_{n+1}c_{n+2})$ be
    connected. Then, for $a_{n+2}\in\cA$ with $S^\cA_{d(\tau)}(\bar a
    a_{n+1} a_{n+2})\cong\tau$, we have
    $\dist^\cA(a_{n+1},a_{n+2})=\dist^\tau(c_{n+1},c_{n+2})\le
    2d(\tau)-1<2d(\tau)$ and therefore $a_{n+2}\in
    B^\cA_{2d(\tau)}(a_{n+1})$. Hence
    \begin{align*}
      |\{a_{n+2}\in\cA\mid S^\cA_{d(\tau)}(\bar a a_{n+1}
      a_{n+2})\cong\tau\}|
      &=|\{a_{n+2}\in B^\cA_{2d(\tau)}(a_{n+1})\mid S^\cA_{d(\tau)}(\bar a a_{n+1} a_{n+2})\cong\tau\}|\\
      &=|\{c\in B^{\tau'}_{2d(\tau)}(c_{n+1})\mid S^\cA_{d(\tau)}(\bar
      c c_{n+1} c)\cong\tau\}|=p
    \end{align*}
    where the last equality follows from $S^\cA_e(\bar a
    a_{n+1})\cong\tau'$ and $e\ge 3d(\tau)$. Hence we showed
    \begin{align*}
      (\cA,\bar a a_{n+1})\models\alpha
         & \iff (\cA,\bar a a_{n+1})\models \exists^{\ge m} x_{n+2}:\sphere_\tau(\bar x,x_{n+1})\\ 
         & \iff p\ge m\\
         & \iff(\cA,\bar a)\models\alpha'\,.
    \end{align*}
  \item Next consider the case that
    $S^{\tau}_{d(\tau)}(c_{n+1}c_{n+2})$ is not connected. Then, for
    $a_{n+2}\in\cA$, we have $S^\cA_{d(\tau)}(\bar a a_{n+1}
    a_{n+2})\cong\tau$ if and only if
    \[
       \dist^\cA(a_{n+1},a_{n+2})\ge 2d(\tau)\text{ and }
       S^\cA_{d(\tau)}(\bar a a_{n+2})\cong\sigma\,.
    \]
    But this implies
    \begin{align*}
      |\{a_{n+2}\in\cA\mid &S^\cA_{d(\tau)}(\bar a a_{n+1} a_{n+2})\cong\tau\}|  \\
         &= |\{a_{n+2}\in\cA\mid \dist^\cA(a_{n+1},a_{n+2})\ge 2d(\tau),S^\cA_{d(\tau)}(\bar a a_{n+2})\cong\sigma\}|  \\
         &= |\{a_{n+2}\in\cA\mid S^\cA_{d(\tau)}(\bar a a_{n+2})\cong\sigma\}|  \\
              &\qquad- |\{a_{n+2}\in\cA\mid \dist^\cA(a_{n+1},a_{n+2})< 2d(\tau),S^\cA_{d(\tau)}(\bar a a_{n+2})\cong\sigma\}|  \\
         &= |\{a_{n+2}\in\cA\mid S^\cA_{d(\tau)}(\bar a a_{n+2})\cong\sigma\}|  \\
              &\qquad- |\{c\in\tau'\mid \dist^{\tau'}(c_{n+1},c)< 2d(\tau),S^{\tau'}_{d(\tau)}(\bar c c)\cong\sigma\}|  \\
         &= |\{a_{n+2}\in\cA\mid S^\cA_{d(\tau)}(\bar a a_{n+2})\cong\sigma\}|  
              - p.
    \end{align*}
    Hence
    \begin{align*}
      (\cA,\bar a a_{n+1})\models\alpha
        &\iff (\cA,\bar a a_{n+1})\models\exists^{\ge m}x_{n+2}:\sphere_\tau\\
        &\iff |\{a_{n+2}\in\cA\mid S^\cA_{d(\tau)}(\bar a a_{n+1} a_{n+2})\cong\tau\}|  \ge m\\
        &\iff |\{a_{n+2}\in\cA\mid S^\cA_{d(\tau)}(\bar a a_{n+2})\cong\sigma\}|  \ge m+p\\
        &\iff (\cA,\bar a)\models\alpha'\,.
    \end{align*}
  \end{enumerate}

  We next evaluate the size of the formula $\psi$. Since $\psi$ is a
  disjunction of formulas $\psi_{\tau'}$, we first fix some $e$-sphere
  $\tau'$ with $n+1$ centers (with $e=3d$). Then $\tau'$ has $\le
  f^{3d-1}\cdot(n+1)$ elements. Hence the formula $\sphere_{\tau'}$
  has size $\le (f^{3d-1}\cdot(n+1))^{O(1)}$ (the constant $O(1)$
  depends on the signature $L$. Now we deal with the formula
  $\varphi_{\tau'}$. It results from $\varphi$ by the replacement of
  subformulas of the form $\alpha=\exists^{\ge
    m}x_{n_2}:\sphere_\tau$. In the first case,
  $|\alpha'|\le|\alpha|$. In the second case, note that $\sigma$ is a
  subsphere of $\tau$, so
  $|\sphere_\sigma|\le|\sphere_\tau|<|\alpha|$. Furthermore, $p\le
  f^{2d(\tau)-1}\le f^{2d-1}$. Note that the formula
  \[
    \alpha'=\exists^{\ge m+p} x_{n_2}:\sphere_\sigma(\bar x,x_{n+2})
  \]
  is shorthand for
  \begin{equation}
     \exists y_1,y_2,\dots, y_{m+p}:
      \bigwedge_{1\le i<j\le m+p}y_i\neq y_j\land
      \forall y\left(\left(\bigvee_{1\le i\le m+p} y=y_i\right)
                              \to \sphere_\sigma(\bar x,y)\right)\,.
      \label{eq:shorthand}
  \end{equation}
  The size of this formula is bounded by
  \[
     O(p^2)+|\sphere_\sigma|\le O(f^{4d-2})+|\alpha|\,.
  \]
  Since $\varphi_{\tau'}$ is obtained from $\varphi$ by at most
  $|\varphi|$ replacements, we obtain
  \[
    |\varphi_{\tau'}|= |\varphi|\cdot O(f^{4d-2})+
           |\varphi|\le |\varphi|\cdot f^{O(d)}
  \]
  and therefore
  \[
   |\varphi_{\tau'}\land\exists^{\ge1} x_{n+1}:\sphere_{\tau'}|
     = f^{O(d)}\cdot(|\varphi|+ n^{O(1)})\,.
  \]

  The number of disjuncts of $\psi$ equals the number of $3d$-spheres
  with $n+2$ centers. Since any such sphere has at most
  $f^{3d-1}(n+1)$ elements, the number of these spheres is bounded by
  \[
     2^{(n\cdot f^{3d-1})^{O(1)}}= 2^{n^{O(1)}\cdot f^{O(d)}}
  \]
  which finally results in
  \[
    |\psi| \le 2^{n^{O(1)}\cdot f^{O(d)}}\cdot
               f^{O(d)}\cdot(|\varphi|+ n^{O(1)})\le 2^{n^{O(1)}\cdot f^{O(d)}}\,.
  \]

  We finally come to the evaluation of the time needed to
  compute~$\psi$. The crucial point in our estimation is the time
  needed to compute the numbers $p$ in (a) and (b); we only
  discuss~(a).

  There are $\le f^{2d(\tau)+1}-1$ candidates $c$ in
  $B^{\tau'}_{2d(\tau)}(c_{n+1})$. For any of them, we have to compute
  the set $B^{\tau'}_d(c)$ (which can be done in time
  $f^{2d+1}-1$). Then, isomorphism of $\tau$ and $S^{\tau'}_d(\bar c
  c_{n+1} c)$ has to be decided. But these are two structures of degree
  $\le f$ and of size $(n+2)\cdot (f^{d+1}-1)\le|\varphi|\cdot
  (f^{d+1}-1)$. Hence, by \cite{Luk82}, this isomorphism test can be
  performed in time polynomial in the size of the structures (the
  degree of the polynomial depends on~$f$). Hence, the number $p$ can
  indeed be computed within the given time bound.  \qed
\end{proof}

We now come to the proof of the central result of this paper:

\begin{proof}[of Theorem~\ref{T-construction}]
  The proof is carried out by induction on the construction of the
  formula $\Phi$. So first, let $\varphi$ be a quantifier-free
  subformula of $\Psi$ whose free variables are among $x_1,\dots,x_n$.
  Let $T$ be the set of all 1-spheres $\tau$ of degree $\le f$ with
  $n+1$ centers such that the constants $c_1,\dots,c_n$ of $\tau$
  satisfy $\varphi$. Then set
  \[
     \psi=\bigvee_{\tau\in T}\exists^{\ge1} x_{n+1}:\sphere_\tau\,.
  \]
  Note that any $1$-sphere with $n+1$ centers has precisely $n+1$
  elements. Furthermore, $n\le|\Phi|$ since $\varphi$ is a subformula
  of $\Phi$. Hence the formula $\sphere_\tau$ has size
  $n^{O(1)}\le|\Phi|^{O(1)}$ and there are $2^{|\Phi|^{O(1)}}$
  disjuncts in the formula $\psi$ (where the constants $O(1)$ depend
  on the signature $L$), i.e., $|\psi|= 2^{|\Phi|^{O(1)}}$.

  We now come to the induction step. The computation of Hanf normal
  forms of $\neg\varphi$ and of $\varphi\lor\varphi'$ are
  straightforward from Hanf normal forms of $\varphi$ and
  $\varphi'$. The only critical point in the induction are subformulas
  of the form $\exists x_{n+1}:\beta$. By the induction hypothesis,
  $\beta$ can be transformed into an $f$-equivalent Hanf normal form
  $\varphi$ and then Lemma~\ref{L-construction} is invoked yielding an
  $f$-equivalent Hanf normal form for $\exists x_{n+1}:\beta$. We have
  to invoke Lemma~\ref{L-construction} at most $|\Phi|$ times where
  the number $n$ is always bounded by $|\Phi|$. Each invokation
  increases the radius of the spheres considered by a factor of three,
  so the maximal radius will be $3^{|\Phi|}=2^{O(|\Phi|)}$. Hence,
  each invokation of Lemma~\ref{L-construction} increases the formula
  by a factor of $2^{f^{2^{O(|\Phi|)}}}$. Putting this to the power of
  $|\Phi|$ does not change the expression.\qed
\end{proof}

\section{Optimality}

In this section, we give a matching lower bound for the size of an
$f$-equivalent formula in Hanf normal form. Namely, we prove

\begin{theorem}
  There is a family of sentences $(\chi_n)_{n\in\bN}$ such that
  $|\chi_n|\in O(n)$ and every 3-equivalent formula $\psi_n$ in Hanf
  normal form has size at least $2^{2^{2^n+1}-1}$.
\end{theorem}

The formulas $\chi_n$ will speak about labeled trees. More formally,
our signature $L$ consists of two binary relations $S_0$ and $S_1$ and
one unary relation~$U$. A structure $\cA=(A,S_0^\cA,S_1^\cA,U^\cA)$
over this signature is a \emph{tree} if there is $X\subseteq\{0,1\}^*$
finite, nonempty and prefix-closed such that
\[
  \cA\cong(X,\{(u,u0)\mid u0\in X\},\{(u,u1)\mid u1\in X\},H)
\]
for some $H\subseteq X$. The tree is \emph{complete} if every inner
node has two children and any two maximal paths have the same length,
this length is called the \emph{height} of the tree (i.e., $X=2^{\le
  h}$ where $h$ is the height). A \emph{forest} is a disjoint union of
trees. As in \cite[Lemma~23]{FriG04}, one can construct formulas
$\chi_n$ of size $O(n)$ such that for any forest~$\cA$, we have
\begin{quote} $\cA\models\chi_n$ if and only if any two complete trees
  of height $2^n$ in $\cA$ are non-isomorphic.
\end{quote}

\begin{lemma}
  Let $\psi$ be a formula in Hanf normal form that is $3$-equivalent
  to $\chi_n$. Then $|\psi|\ge2^{2^{2^n+1}-1}$.
\end{lemma}

\begin{proof}
  Suppose $|\psi|<2^{2^{2^n+1}-1}$. Let $M$ be the maximal number $m$
  such that $\exists^{\ge m}x:\sphere_\sigma$ appears in $\psi$ (for
  any sphere $\sigma$). We can assume that $\psi$ does not contain any
  formula $\sphere_\sigma$ where $\sigma$ is a $1$-sphere.  The
  complete tree of height $2^n$ has $2^{2^n+1}-1$ nodes. Hence there
  are $2^{2^{2^n+1}-1}$ ways to color such a tree. Since we assume
  $|\psi|<2^{2^{2^n+1}-1}$, there is one such tree~$\cB$ (with
  root~$r$) such that the formula $\sphere_{(\cB,r)}$ does not appear
  in~$\psi$.

  Next, we need a bit of terminology. If $\cA$ is a tree, $a$ a node
  in $\tau$, and $d\in\bN$, then also $\tau\restriction B^\cA_d(a)$ is
  a tree that we denote~$N^\cA_d(a)$. Recall that the sphere
  $S^\cA_d(a)=(N^\cA_d(a),a)$ about $a$ of radius~$d$ has an
  additional constant.

  Now we define a structure $\cA_0$. It consists of $M+1$ copies of
  any of the structures $N^\cB_d(b)$ where
  \begin{enumerate}
  \item $1<d\le 2^n$ and $b$ is not the root of $\cB$ or
  \item $d<2^n$.
  \end{enumerate}

  Finally, let $\cA_2=\cA_0\uplus\cB\uplus\cB$ be the disjoint union
  of $\cA_0$ and two copies of the tree~$\cB$. Since $\cA_0$ does not
  contain any complete tree of height $2^n$, we get
  $\cA_0\models\chi_n$ and therefore $\cA_0\models\psi$. Note that any
  sphere realized in $\cA_0$ or $\cA_2$ is also realized in $\cB$. So
  let $b\in\cB$, and $d\in\bN$. We distinguish several cases:
  \begin{enumerate}
  \item $1<d\le 2^n$ and $b$ is not the root of $\cB$. Then the
    sphere $(N^\cB_d(b),b)$ is realized in $\cA_0$ more than $M$
    times, hence the same holds for $\cA_2$.
  \item $d<2^n$. Then $(N^\cB_d(b),b)$ is realized in $\cA_0$ more
    than $M$ times, hence the same holds for~$\cA_2$.
  \item $b$ is the root of the tree $\cB$ and $d=2^n$. Then
    $N^\cB_d(b)=\cB$. Hence $S^\cB_d(b)$ is not realized in~$\cA_0$
    and it is realized twice in~$\cA_2$. But validity of $\psi$ does
    not depend on this number since $\psi$ does not mention the
    formula $\sphere_{(\cB,b)}$.
  \end{enumerate}
  Hence, we obtain $\cA_2\models\psi$, contrary to our assumption that
  $\chi_n$ and $\psi$ are $3$-equivalent.\qed
\end{proof}

The theorem now follows immediately from this lemma.


\end{document}